\documentclass[a4paper,11pt]{article}
\usepackage{fullpage}
\usepackage{graphicx}
\usepackage{amsmath}
\usepackage{amsthm} 
\usepackage{dsfont}
\usepackage{amssymb}
\usepackage{array}
\usepackage{multicol}
\usepackage[latin1]{inputenc}
\usepackage[T1]{fontenc}
\usepackage{MnSymbol}
\usepackage{multicol}
\usepackage{tikz}
\usepackage{tikz-cd}
\usetikzlibrary{arrows}
\usetikzlibrary{decorations.pathmorphing}

\newtheorem{de}{Definition}[section]
\newtheorem{theo}{Theorem} [section] 
\newtheorem{prop}{Proposition}[section]
                                   
\newtheorem{lem}{Lemma}[section]
\newtheorem{rem}{Remarks}[section] 
\newtheorem{coro}{Corollary}[section]

\newcommand{\Hcal}{\mathcal{H}}
\newcommand{\Kcal}{\mathcal{K}}

\newcommand{\Lcal}{\mathcal{L}}
\newcommand{\Pcal}{\mathcal{P}}
\newcommand{\Acal}{\mathcal{A}}

\newcommand{\Ecal}{\mathcal{E}}
\newcommand{\Bcal}{\mathcal{B}}

\newcommand{\Mcal}{\mathcal{M}}
\newcommand{\Gcal}{\mathcal{G}}
\newcommand{\Prob}{\mathbb{P}}
\newcommand{\Fcal}{\mathcal{F}}

\newcommand{\R}{\mathbb{R}}
\newcommand{\Ebb}{\mathbb{E}}
\newcommand{\Nbb}{\mathbb{N}}

\newcommand{\C}{\mathbb{C}}

\newcommand{\Sp}{\text{Sp}}
\newcommand{\Ind}{\mathds{1}}

\newcommand{\sca}[2]{\langle #1 ,#2\rangle}
\newcommand{\norm}[1]{\left\| #1 \right\|}
\newcommand{\proj}[2]{|#1\rangle\langle #2|}
\newcommand{\eps}{\varepsilon}
\newcommand{\Tr}{\text{Tr}}
\newcommand{\off}{\text{off}}

\numberwithin{equation}{section}

\author{\vspace{0.5cm}
				Ivan Bardet\footnote{Institut Camille Jordan, Universit\'e Claude Bernard Lyon 1, 43 boulevard du 11 novembre 1918, 69622 Villeurbanne, cedex France
}\\
}
\date{}
\title{Quantum extensions of dynamical systems and of Markov semigroups\thanks{Work supported by ANR-14-CE25-0003 "StoQ" }}

\begin{document}

\maketitle

\begin{abstract}
We investigate some particular completely positive maps which admit a stable commutative Von Neumann subalgebra. The restriction of such maps to the stable algebra is then a Markov operator. In the first part of this article, we propose a recipe in order to find a quantum extension of a given Markov operator in the above sense. We show that the existence of such an extension is linked with the existence of a special form of dilation for the Markov operator studied by Attal in \cite{Att1}, reducing the problem to the extension of dynamical system. We then apply our method to the same problem in continuous time, proving the existence of a quantum extension for L\'evy processes. In the second part of this article, we focus on the case where the commutative algebra is isomorphic to $\Acal=l^\infty(1,...,N)$ with $N$ either finite or infinite. We propose a classification of the CP maps leaving $\Acal$ stable, producing physical examples of each classes.
\end{abstract}

\section{Introduction}

We are interested in those quantum dynamics, on the algebra of bounded operators on some separable Hilbert space, that admit a stable commutative subalgebra. The interest of such a property holds in the fact that if this algebra is maximal then it can be looked at as the algebra of bounded functions on some measurable space. The restriction to this subalgebra is a semigroup of positive and identity preserving operators, that is, a classical Markov semigroup. Consequently, the evolution of the observables of this algebra follows a classical evolution. For this reason, we call such dynamics \emph{subclassical}.
\\

The question of the existence of such a stable algebra was first arised and motivated by Rebolledo in \cite{Reb2} and then studied by several authors (\cite{F-S}, \cite{B-F-S}). We shall be concerned with the inverse problem: given a classical Markov semigroup, is it the restriction of a quantum Markov semigroup to a maximal stable commutative algebra? Answering this question would help understanding which classical processes can appear in quantum dynamics.

There has already been a lot of works on the existence of such quantum extensions. The first motivation was to show the possibility to embed the classical theory of stochastic calculus in the quantum one, created by Hudson and Parthasarathy in their pioneering article \cite{H-P}. Such an embedding shows that the latter is a true generalization of the former. This fact was already remarked by Meyer in \cite{Mey3}, in the case of finite Markov chains in discrete time. Meyer's construction was extended to the case of jump processes by Parthasarathy and Sinha in \cite{P-S} who constructed the structure maps of the flow through certain group actions. This was followed by the proof of the existence of quantum extensions for Az\'ema martingales \cite{C-F} by Chebotarev and Fagnola, and Bessel processes \cite{F-M} by Fagnola and Monte. The main idea behind the two last results was highlighted by Fagnola in \cite{Fagn}, where he gives sufficient conditions on the generator of a Quantum Markov Semigroup (QMS) in order to admit a classical restriction. As an application, he proved the existence of a quantum extension for diffusive Markov semigroups. Basically, it amounts to prove that the generator is itself a quantum extension of the generator of a classical Markov semigroup. The subsequent recipe has the great advantage to avoid technical problems such as the continuity of the QMS. The main difficulty consists in finding the adequate quantum generator.
\\

In this article, we propose a different recipe to find quantum extension of a Markovian dynamics. This recipe can be view as a generalization of one used by Gregoratti in \cite{Gr2}. It is based on the existence of a particular form of classical dilation of a Markov operator. In fact, it was shown in \cite{Att1} by Attal that Markov operators on Lusin space all admit dilations in the form of a dynamical system on a product space. This particular case of classical dilations was further studied by Deschamps \cite{Des1}, who showed how their limits from discrete to continuous time lead to a specific class of stochastic differential equations.
\\Under some basic assumptions, we make explicit quantum unitary extensions of those dynamical systems. We show that the quantum conditional expectation (the trace) of this unitary evolution over a state of the environment is a quantum extension of the primary Markov operator.
\\
Consequently, the main difficulty of our recipe lies in the existence of a classical dilation in the Attal-Gregoratti sense, that fulfills the condition to be extended. Focusing on the continuous-time case, we find sufficient conditions on a Markov semigroup to have such a dilation. Although they are only sufficient and not necessary conditions, we believe that they can highlight the probabilistic nature of the problem. Finally, putting altogether these conditions leads naturally  to the proof of the existence of a quantum extension for L\'evy processes.
\\

In practice, the stable algebra that appears in physical examples is the algebra generated by a self-adjoint operator of the Hilbert space, i.e. an observable of the system. One of the most famous example is obtained in the weak coupling or Van-Hoove limit, that was put in a rigorous mathematical language by Davies \cite{Dav1} (see also \cite{A-K}). There has already been a lot of studies on the QMS that arise in this limit (\cite{A-F-H} for example), called generic QMS. An other example of physically realizable Completely Positive map (CP map) with a stable commutative algebra is given by the composition with the Von Neumann measurement operator of any CP map. Those examples appear in the case where the observable generating the algebra has discrete and non-degenerate spectrum. Considering these particular cases, we provide a classification of subclassical CP maps. We also highlight a link between subclassical CP maps and quantum trajectories.
\\

This article is structures as followed. In Section \ref{sectdef} we introduce our notations and our definitions. The main results about quantum extensions are in Section \ref{sectquantext}, where we explain our recipe to find quantum extensions of classical dynamics, and apply it to L\'evy processes. Finally in Section \ref{sectdisret}, we focus on the simpler case where the stable commutative algebra is isomorphic to $\Acal=l^\infty(1,...,N)$ with $N$ either finite or infinite.

\label{sectintro}
\section{Notations and definitions}

In this section we give the set-up of our study. It is mainly the same set-up as described by Rebolledo in \cite{Reb1}, however we restrict ourselves to the case where the commutative algebra $\Acal$ is a Von Neumann algebra (and not only a $C^*$ algebra). In Subsection \ref{subsectdef1}, we recall the definitions of Markov quantum dynamics for open system and the kind of commutative subalgebra we will consider in the following. In Subsection \ref{subsectdef2}, we give our definitions of quantum dynamics with stable commutative subalgebra.

\label{sectdef}
\subsection{Quantum probability background for open quantum systems}

We are interested in quantum dynamics acting on the algebra $\Bcal(\Hcal)$ of bounded operator on a \emph{separable} Hilbert space $\Hcal$. This algebra physically corresponds to the set of the observables of a quantum system $\Hcal$. When the system is closed, in the Heisenberg point of view, evolution of observables is given by a group of $*$-homomorphism $(U_t^*\cdot U_t)_{t\geq0}$, where $(U_t)_t\geq0$ is a one parameter group of unitary operators. In this case, we shall talk about a \emph{quantum unitary evolution}.
\\However, in this article we focus on open quantum systems. Quantum Markov process (QMS) are defined as weakly*-continuous semigroups $(\Pcal_t)_{t\geq0}$ of completely positive maps (CP map) on $\Bcal(\Hcal)$, such that $\Pcal_t(\Ind)=\Ind$. Either the time is discrete or continuous, we will talk about quantum dynamics or simply dynamics on $\Bcal(\Hcal)$.
\\
\\\textit{In the following, all CP maps on $\Bcal(\Hcal)$ are normal (that is continuous for the $\sigma$-weak topology on $\Bcal(\Hcal)$) and identity preserving}.
\\

We shall be concerned with quantum dynamics that have a commutative Von Neumann subalgebra of $\Bcal(\Hcal)$ stable under there action.

\begin{de}
Let $\Acal$ be a commutative subalgebra of $\Bcal(\Hcal)$. We say that $\Acal$ is \emph{maximal} if it admits a cyclic vector i.e. a vector $\Omega\in\Hcal$ such that $\Acal\Omega$ is dense in $\Hcal$.
\label{demaximal}
\end{de}

\noindent If $(E,\Ecal,\mu)$ is a measured space and $f\in L^\infty(E,\Ecal,\mu)$ (we simply write $L^\infty(\mu)$), the multiplication operator $M_f$ on $L^2(\mu)$ is the bounded operator defined by

\begin{equation*}
\left(M_f\ g\right)(x)= f(x)g(x),\ g\in L^2(\mu),\ x\in E.
\end{equation*}

\begin{prop}[See \cite{R-S}]
Let $\Acal$ be a commutative subalgebra of $\Bcal(\Hcal)$. The following three assertions are equivalent:
\begin{itemize}
\item[(i)]$\Acal$ is maximal;
\item[(ii)]$\Acal$ is unitarily equivalent to the algebra of multiplication operators on a Hilbert space $L^2(E,\Ecal,\mu)$ for some measured space $(E,\Ecal,\mu)$;
\item[(iii)]$\Acal=\Acal'$, where $\Acal'$ stands for the commutant of $\Acal$.
\end{itemize}
\label{propvectcyclic}
\end{prop}

\noindent This proposition tells us that whenever $\Acal$ is maximal, the Hilbert space $\Hcal$ and the algebra $\Acal$ can be identified respectively with $L^2(E,\Ecal,\mu)$ for some measured space $(E,\Ecal,\mu)$ and the algebra $L^\infty(E,\Ecal,\mu)$ of bounded functions on this measured space. Moreover, in the corresponding physical applications, the algebra which is stable under the evolution of the system is often $\Acal'$. This algebra is commutative if and only if it is equal to $\Acal$. Thus the assumption that $\Acal$ is maximal is also needed in order to obtain a classical interpretation of the physically stable algebra.
\\
\\\emph{In the following we will always assume that $\Acal$ is maximal, and consequently has a cyclic vector $\Omega$, which we choose to be norm $1$.}

\label{subsectdef1}
\subsection{Preliminary definitions}

We begin with the definition of \emph{subclassical dynamics}. The system is modeled by the Hilbert space $\Hcal=L^2(E,\Ecal,\mu)$ for some measured space $(E,\Ecal,\mu)$. The commutative subalgebra $\Acal$ of $\Bcal(\Hcal)$ is the algebra of multiplication operators by bounded functions on $E$:

\begin{equation}
\Acal=\left\{M_f,\ f\in L^\infty(\mu)\right\}.
\label{eqalgmultop}
\end{equation}

\noindent Thus $\Acal$ can be identified with $L^\infty(\mu)$.

\begin{de}[Subclassical CP map]
A completely positive map $\Lcal$ on $\Bcal(\Hcal)$ is called \emph{$\Acal$-subclassical} if $\Acal$ is stable under the action of $\Lcal$, i.e.
\begin{equation*}
\Lcal(\Acal)\subset\Acal.
\end{equation*}
\label{defSCCP}
\end{de}

\begin{de}[Subclassical QMS]
A quantum Markov semigroup $(\Pcal_t)_{t\geq0}$ is called \emph{$\Acal$-subclassical} if $\Acal$ is a stable subalgebra of $(\Pcal_t)_{t\geq0}$, i.e.
\begin{equation*}
\Pcal_t(\Acal)\subset\Acal \text{ for all }t\geq0.
\end{equation*}
\label{descQMS}
\end{de}

\noindent Note that if $\Lcal$ is a $\Acal$-subclassical CP map, there obviously exists a linear operator $L$ acting on $L^\infty(\mu)$ such that for all $f\in L^\infty(\mu)$:

\begin{equation*}
\Lcal(M_f)=M_{Lf}.
\end{equation*}

\noindent If furthermore $\Lcal$ is an identity preserving and normal CP map, it is easy to verify that $L$ is a Markov operator on $L^\infty(\mu)$, that is, it satisfies:

\begin{itemize}
\item[i)]$L$ is an operator on $L^\infty(\mu)$;
\item[ii)]$L\Ind=\Ind$;
\item[iii)]$Lf\geq0$ whenever $f\geq0$;
\item[iv)]$f\mapsto Lf$ is $\sigma$-weakly continuous.
\end{itemize}

\noindent In the same way, if $(\Pcal_t)_{t\geq0}$ is a $\Acal$-subclassical normal QMS, then there exists a Markov semigroup $(P_t)_{t\geq0}$ on $L^\infty(\mu)$, i.e. $P_t$ are Markov operators for all $t\geq0$, $P_0f=f$ on $L^\infty(\mu)$, $t\mapsto P_t$ is strongly continuous on $L^\infty(\mu)$ and for all $f\in L^\infty(\mu)$:

\begin{equation*}
\Pcal_t(M_f)=M_{P_tf}\text{ for all }t\geq0.
\end{equation*}
\\

In the next section, we focus on the following problem: given a Markov dynamics, either in discrete or continuous time, is it the restriction of a quantum dynamics on a stable commutative algebra? We shall use the following definition, which naturally extend to the continuous case.

\begin{de}
Let $(E,\Ecal,\mu)$ be a probability space, and let $L$ be a Markov operator on $L^\infty(\mu)$. A CP map $\Lcal$ on $\Bcal(L^2(\mu))$ is called a \textbf{quantum extension} of $L$ if for all $f\in L^\infty(\mu)$,
\begin{equation}
\Lcal(M_f)=M_{Lf},
\label{eqdefquantumext}
\end{equation}
\noindent i.e. $\Lcal$ is $\Acal$-subclassical with $\Acal=\{M_f,\ f\in L^\infty(\mu)\}$ and the classical restriction of $\Lcal$ to $\Acal$ is $L$.
\label{dequantumext}
\end{de}

\label{subsectdef2}
\section{Subclassical dynamics: the dynamical system point of view}

In this section we show how quantum extensions of Markov operators can be obtained by a general scheme. This scheme relies on the following remarks:

\begin{itemize}
\item[1)] it is easy to give conditions in order to find quantum extensions of dynamical systems;
\item[2)] dynamical systems naturally appear as dilation of Markov operators in a certain form;
\item[3)] finally one can recover a quantum extension of a Markov operator via the quantum extension of its dilation.
\end{itemize}

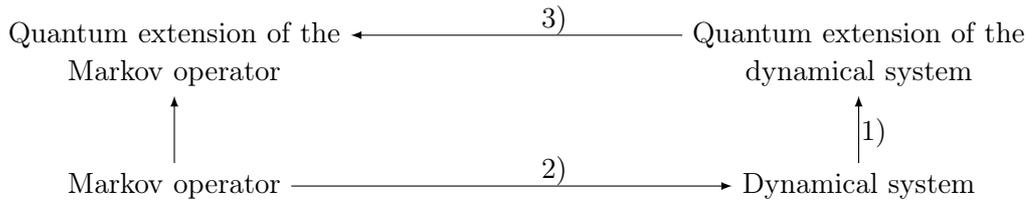
\begin{tikzpicture}
\node (Q1) at (-3,1) {Quantum extension of the};
\node (Q2) at (-3,0.5) {Markov operator};
\node (P1) at (6,1) {Quantum extension of the};
\node (P2) at (6,0.5) {dynamical system};
\node (R) at (-3,-1) {Markov operator};
\node (L) at (6,-1) {Dynamical system};
\node at (6.2,-0.3) {1)};
\node at (2,1.2) {3)};
\node at (2,-0.8) {2)};
\draw[->,>=latex] (P1)--(Q1);
\draw[->,>=latex] (R)--(Q2);
\draw[->,>=latex] (L)--(P2);
\draw[->,>=latex] (R)--(L);
\end{tikzpicture}

In Subsection \ref{subsectquantext1}, we focus on point 1). We find conditions so that a dynamical system admits a quantum \emph{unitary} extension and under these conditions we completely characterize such possible extensions. In Subsection \ref{subsectquantext2}, we explain point 2) and 3) and we prove that the quantum extension does not depend on the choice of the classical dilation in Subsection \ref{subsectquantext3}. In Subsection \ref{subsectquantext4} we apply our recipe for the case of finite states spaces.
\\Finally in Subsection \ref{subsectquantext5} we treat the case of continuous time Markov semigroups, which allows us to prove the existence of a quantum extension for L\'evy processes in Subsection \ref{subsectquantext6}.

\label{sectquantext}
\subsection{Subclassical unitary evolutions as quantum extensions of dynamical systems}

In this section we look for conditions on a dynamical system that allow the existence of its quantum extension. We first emphasize in Proposition \ref{propevolunitaire} that unitary evolutions are natural candidates for quantum extension of dynamical system. This leads us to find appropriate conditions on a dynamical system to admits a quantum extension given by a unitary evolution (that is, a quantum unitary extension). That is the aim of Proposition \ref{propunitact}. Finally in Proposition \ref{propunicitesubclassunitary} we characterize all quantum unitary extensions.
\\

A dynamical system is a quadruplet $(E,\Ecal,\mu,T)$, where $(E,\Ecal,\mu)$ is some measured space, and $T:E\to E$ is a $\Ecal$-measurable function. It can be seen as a very particular Markov operator. Indeed, we will say that a Markov operator $L$ on the measured space $(E,\Ecal,\mu)$ is deterministic if there exist a measurable function $T$ on $(E,\Ecal,\mu)$ such that for all $f\in L^\infty(\mu)$, we have:

\begin{equation*}
Lf=f\circ T.
\end{equation*}

\begin{prop}
Suppose $(E,\Ecal,\mu)$ is a Lusin space. Then the restriction of an $\Acal$-subclassical unitary evolution on $L^2(\mu)$ to its classical part is deterministic.
\label{propevolunitaire}
\end{prop}

\begin{proof}
Let $U$ be a unitary operator on $L^2(\mu)$. We suppose that the corresponding unitary evolution is $\Acal$-subclassical, i.e. there exists a Markov operator $L$ on $L^\infty(\mu)$ such that for all $f\in L^\infty(\mu)$, 

\begin{equation*}
U^*M_fU=M_{Lf}.
\end{equation*}

\noindent For all $f,g\in L^\infty(\mu)$, we have $M_{Lf}M_{Lg}=M_{Lfg}$ and $M_{L\overline{f}}=M_{\overline{Lf}}$. Consequently, $L$ is a $*$-homomorphism of the algebra $L^\infty(\mu)$. As was shown by Attal in \cite{Att1}, under the hypothesis that $(E,\Ecal,\mu)$ is a Lusin space, Markov operators that are $*$-homomorphism of the algebra $L^\infty(\mu)$ are exactly the deterministic ones.
\end{proof}

This proposition tells us that if $U$ is a unitary operator on $\Hcal$ and if the corresponding evolution on $\Bcal(\Hcal)$ is $\Acal$-subclassical, then there exists a measurable application $T:E\to E$ such that for every $f\in L^\infty(\mu)$,

\begin{equation}
U^*M_fU=M_{f\circ T},
\label{eqlienUnitsystdyn}
\end{equation}

\noindent i.e. $U^*\cdot U$ is a quantum extension of the dynamical system $(E,\Ecal,\mu,T)$. Consequently from Proposition \ref{propevolunitaire} unitary evolutions appears as natural candidates for quantum extensions of dynamical systems. This remark is made effective in the next proposition.

\begin{prop}
Let $(E,\Ecal,\mu,T)$ be a dynamical system. Assume that $T$ is invertible and $\mu$-preserving. Then the operator $U_T:L^2(\mu)\to L^2(\mu)$ defined by
\begin{equation}
U_Tf=f\circ T,
\label{eqpropunitact1}
\end{equation}
\noindent is a unitary operator, and we have $U_T^*=U_{T^{-1}}$. Moreover, for all $f\in L^\infty(\mu)$, we have
\begin{equation}
U_TM_fU_{T^{-1}}=M_{f\circ T}.
\label{eqpropunitact2}
\end{equation}
\label{propunitact}
\end{prop}

\begin{proof}
The first part of the proposition is the well-known Koopman's Lemma (see \cite{R-S}). We prove the second part.
For all $f\in L^\infty(\mu)$ and $g\in L^2(\mu)$ we have:
\begin{align*}
U_TM_fU_{T^{-1}}g
& = U_T M_f (g\circ T^{-1}) \\
& = U_T \left[f \left(g\circ T^{-1}\right)\right] \\
& = [f (g\circ T^{-1})]\circ T \\
& = (f\circ T) g \\
& = M_{f\circ T}g.
\end{align*}
\end{proof}

\noindent Thus, under the condition that $T$ is invertible and $\mu$-preserving, the unitary evolution given by $U_{T^{-1}}$ is a quantum extension of $T$. From now on, this choice of a quantum extension will be called the \emph{canonical quantum extension}. We end this section by a characterization of all quantum unitary extensions of invertible and measure-preserving dynamical systems.

\begin{prop}
Let $T$ be an invertible and $\mu$-preserving dynamical system on $(E,\Ecal,\mu)$. Let $U=U_{T^{-1}}$ be the canonical quantum extension of $T$. Let $V$ be a unitary operator on $L^2(\mu)$. Then the following two assertions are equivalent.
\begin{itemize}
\item[(i)] For all $f\in L^\infty(\mu)$, $V^*M_fV=M_{f\circ T}$.
\item[(ii)] There exists $g\in L^\infty(\mu)$ such that $V=M_g\ U$.
\end{itemize}
Furthermore, if this is realized, then $g$ in (ii) is such that $|g|^2=1$ $\mu$-almost surely. 
\label{propunicitesubclassunitary}
\end{prop}

\begin{proof}
The last point of the theorem is a straightforward computation.
\\Recall that $\Acal$ is the algebra of multiplication operators by bounded functions. Then (ii) is equivalent to $VU^*\in\Acal=\Acal'$, which itself is equivalent to $VU^*M_f=M_fVU^*$ for all $f\in L^\infty(\mu)$. Multiplying to the right by $V^*$ and to the left by $U$, we get that (ii) is equivalent to $V^*M_fV=U^*M_fU$ for all $f\in L^\infty(\mu)$, which is (i).
\end{proof}

Consequently quantum extensions of an invertible and measure-preserving dynamical system $(E,\Ecal,\mu T)$ are in one-to-one correspondence with elements $g\in L^\infty(\mu)$ of constant modulus $1$.

\label{subsectquantext1}
\subsection{Subclassical CP maps as quantum extensions of classical dilations}

As was shown by Gregoratti in \cite{Gr1} for the case of finite spaces and latter by Attal in \cite{Att1} in the general case, classical dynamical systems appear as natural dilations of Markov operators: the latters can be written as the "trace" of a deterministic evolution on a larger space.
\\As a second step of our recipe, in this section we show that the trace over a state of the environment of a unitary quantum extension leads to a quantum extension of the Markov operator.
\\First we recall a result obtain by Attal in \cite{Att1}.
\\ 

Let $(E,\Ecal,\mu)$ be a measured space, $\mu$ being not necessarily finite. This space stands for the small system. Let $(F,\Fcal,\nu)$ be a probability space, which stands for the environment. Let $T$ be an $\Ecal\otimes\Fcal$-measurable application from $E\times F$ to $E\times F$. The space $(E\times F,\Ecal\otimes\Fcal,\mu\otimes\nu,T)$ is a dynamical system, which describes the evolution of the whole system.
\\
\\In the Heisenberg description of time evolution of a physical system, what we consider is the evolution of observables rather than the evolution of the state of the system. In the commutative case, observables are functions $h\in L^\infty(\mu\times \nu)$ and their evolution is given by the operator $\tilde{T}$ on $L^\infty(\mu\otimes \nu)$ defined by $\Tilde{T}:h\mapsto h\circ T$.
\\
\\Suppose now that we only have access to the space $E$ and the action of the environment is only given by the expectation over the measure $\nu$ on $(F,\Fcal)$. Starting at a deterministic point $x\in E$ and at a random point $y\in F$ with law $\nu$, one step evolution of an observable $f\in L^\infty(\mu)$ is given by the operator $L$ on $L^\infty(\mu)$ given by:

\begin{equation}
Lf(x)=\int_{F}{\tilde{T}(f\otimes\Ind)(x,y)d\nu(y)},\ f\in L^\infty(\mu),\ x\in E.
\label{eqclassdil}
\end{equation}

\begin{prop}[Theorem 2.2 in \cite{Att1}]
The operator $L$ is a Markov operator on $L^\infty(E)$.
\\Conversely, for any Markov operator on a \emph{Lusin} space $(E,\Ecal,\mu)$, one can find a measurable space $(F,\Fcal)$, a probability measure $\nu$ on $F$ and an invertible dynamical system $T:E\times F\to E\times F$ such that $L$ is given by Equation (\ref{eqclassdil}).
\\Thus a Markov dynamics can always be seen as the restriction to a subsystem of a deterministic dynamics. 
\label{propclassdil}
\end{prop}

We now come back to the quantum setup. The small system is the Hilbert space $\Hcal=L^2(E,\Ecal,\mu)$ and the environment is the Hilbert space $\Kcal=L^2(F,\Fcal,\nu)$. The only information we have on the environment is given by a state $\omega$ on $\Kcal$. This state defines on $(F,\Fcal)$ a probability measure $\nu_\omega$ via :

\begin{equation}
\nu_\omega(J)=\Tr[\omega M_{\Ind_J}],
\label{eqmestrace}
\end{equation}

\noindent where $M_{\Ind_J}$ is the multiplication operator by the characteristic function $\Ind_J$ of the set $J\in\Fcal$.

\begin{theo}
Let $T:E\times F\to E\times F$ be a dynamical system. Suppose that there exists a quantum extension of $T$ on $\Hcal\otimes\Kcal$ given by a unitary operator $U$ on $\Hcal\otimes\Kcal$, that is, for all $h\in L^\infty(E\times F)$ we have
\begin{equation*}
U^*M_hU=M_{h\circ T}.
\end{equation*}
\noindent Let $\omega$ be a state on $\Kcal$. Then the application $\Lcal_\omega:\Bcal(\Hcal)\to\Bcal(\Hcal)$ defined by
\begin{equation}
\Lcal_\omega(X)=\Tr_{\omega}[U^*(X\otimes I)U],
\label{eqtheoQclassdil1}
\end{equation}
is a $\Acal$-subclassical CP map on $\Bcal(\Hcal)$, such that the associated Markov operator $L_\omega$ is given by
\begin{equation}
L_\omega f(x)=\int_{F}{(f\otimes\Ind)\circ T(x,y)d\nu_\omega(y)},\ f\in L^\infty(E).
\label{eqtheoQclassdil2}
\end{equation}
\label{theoQclassdil}
\end{theo}

\begin{proof}
The proof is almost straightforward. The fact that $\Lcal_\omega$ is a normal identity preserving CP map is well-known. Consequently we just have to check that, for all $f\in L^\infty(E)$, we have
\begin{equation*}
\Tr_{\omega}[U^*(M_f\otimes I)U]=M_{L_\omega f},
\end{equation*}
\noindent with $L_\omega$ given by Equation (\ref{eqclassdil}). Then, for all $\rho\in \Lcal_1(\Hcal)$,
\begin{align*}
\Tr\left\{\Tr_{\omega}[U^*(M_f\otimes I)U]\rho\right\}
& = \Tr\left\{M_{T(f\otimes\Ind)} (\rho\otimes\omega)\right\} \\
& = \int_{E\times F}{T(f\otimes\Ind)d\mu_\rho\otimes d\nu_\omega} \\
& = \int_E\left[\int_F{T(f\otimes\Ind)d\nu_\omega}\right]d\mu_\rho \\
& = \int_E{[L_\omega f]d\mu_\rho}  \\
& = \Tr[M_{L_\omega f}\rho].
\end{align*}
\noindent This shows that $\Tr_{\omega}[U^*(M_f\otimes I)U]=M_{L_\omega f}$.
\end{proof}

\begin{rem}
For all states $\omega$ on $\Kcal$, the probability measure $\nu_\omega$ is absolutely continuous with respect to $\nu$. Indeed, if $\omega$ is a pure state, i.e. $\omega=\proj{\psi}{\psi}$ for some $\psi\in\Kcal$, then for all $J\in\Fcal$, $\nu_\omega(J)=\sca{\psi}{M_{\Ind_J}\psi}=\int_J{|\psi|^2d\nu}$ and obviously $\nu_\omega\ll\nu$. The same holds if $\omega$ is a mixed state by diagonalizing it as a mixture of pure states.
\\Conversely, if $\tilde{\nu}$ is a probability measure on $(F,\Fcal)$ which is absolutely continuous with respect to $\nu$, then there exists $\psi\in\Kcal$ such that for all $J\in\Fcal$, $\tilde{\nu}=\int_J{|\psi|^2d\nu}$. Consequently taking $\omega=\proj{\psi}{\psi}$ leads to $\nu_\omega=\tilde{\nu}$.
\label{remmeasure}
\end{rem}

In Theorem \ref{theoQclassdil}, the quantum extension of the dynamical system is not necessarily the canonical choice. Furthermore it is not necessary for the dynamical system to be invertible and measure preserving, as soon as it possesses a quantum unitary extension. However, the quantum extension of the Markov operator depends on this choice, as it will not act in a same way outside of $\Acal$. The quantum extension associated to the canonical choice for the unitary operator, when it exists, will be called the \emph{canonical choice associated to $T$}. In the next Subsection, we prove that the canonical choice does not depend on the choice of the dilation $T$.

\label{subsectquantext2}
\subsection{Uniqueness of the quantum extensions}

So far, we have given a general method to construct a quantum extension of a given Markov operator $L$. We have already noticed that this quantum extension depends on the choice of the quantum unitary extension of the dilation. However we are going to prove that it does not depend on the choice of the classical dilation of $L$.

\begin{theo}
Suppose $(E\times F_1,\Ecal\otimes\Fcal_1,\mu\otimes\nu_1,T_1)$ and $(E\times F_2,\Ecal\otimes\Fcal_2,\mu\otimes\nu_2,T_2)$ are both invertible and measure-preserving dilations of $L$. Then the \emph{canonical} quantum extension associated to each dilation are the same.
\label{theominimaldilation}
\end{theo}

\begin{proof}
We define the CP map $\vec{L}$ associated to $L$:
\begin{equation*}
\begin{array}{cccc}
\vec{L}: & L^\infty(\mu) & \to & \Bcal(L^2(\mu)). \\
& f & \mapsto & M_{Lf}
\end{array}
\end{equation*}
Let $(E\times F_i,\Ecal\otimes\Fcal_i,\mu\otimes\nu_i,T_i)$, $i=1,2$, be two classical dilations of $L$ such that the $T_i$'s are invertible and $\mu\otimes\nu_i$-preserving.
\noindent Two Stinespring representations \cite{Fagn} of $\vec{L}$ are given by $(L^2(\mu\otimes\nu_i),V_i)$, where the operators $V_i$ are defined by
\begin{equation*}
\begin{array}{cccc}
V_i: & L^2(\mu) & \to & L^2(\mu\otimes\nu_i)\qquad . \\
& h & \mapsto & (h\otimes\Ind)\circ T_i^{-1}
\end{array}
\end{equation*}
\noindent Then we have for all $f\in L^\infty(\mu)$:
\begin{equation*}
\vec{L}f=V_i^* M_{f}\otimes I_{\Kcal_i} V_i.
\end{equation*}
\noindent As the $T_i$'s are invertible and measure-preserving, the $V_i$'s are isometric. Furthermore, the canonical quantum extension associated to each dilation is
\begin{equation*}
\Lcal_i(\cdot)=V_i^*(\cdot\otimes I)V_i.
\end{equation*}
We can always restrict the Hilbert spaces $L^2(\mu\otimes\nu_1)$ and $L^2(\mu\otimes\nu_2)$ to the closure of the vector fields spanned by the sets $\{M_f\otimes I\ V_1h,\ f\in L^\infty(\mu),\ h\in L^2(\mu)\}$ and $\{M_f\otimes I\ V_2h,\ f\in L^\infty(\mu),\ h\in L^2(\mu)\}$ respectively, so that both Stinespring representations are \emph{minimal}. Consequently (see \cite{Fagn} for example) there exists a unitary operator $W:L^2(\nu_1)\to L^2(\nu_2)$ such that for all $f\in L^\infty(\mu)$ and $h\in L^2(\mu)$
\begin{equation*}
I\otimes W\left(M_f\otimes I\ V_1\ h\right)=M_f\otimes I\ V_2\ h.
\end{equation*}
In particular, $I\otimes W V_1=V_2$ and for all $X\in\Bcal(L^2(\mu))$,
\begin{align*}
\Lcal_1(X) 
& = V_1^* (X\otimes I) V_1 \\
& = V_2^* \left(I\otimes W\ X\otimes I\ I\otimes W\right)V_2\\
& = V_2^* (X\otimes I) V_2\\
& = \Lcal_2(X).
\end{align*}
\end{proof}

\label{subsectquantext3}
\subsection{Quantum unitary extensions of dynamical systems for discrete states spaces}
We are now going to apply Theorem \ref{theoQclassdil} to the case of discrete states spaces.
The existence of a quantum extension for discrete states spaces was already proved by Parthasarathy and Sinha in \cite{P-S} using quantum stochastic calculus. It was further studied by Gregoratti in \cite{Gr2} and \cite{Gr1} as an application of Theorem \ref{theoQclassdil}. In this section we want to emphasize that reversibility of the dilation is not only a sufficient but also a necessary condition. For discrete spaces, the measure-preserving condition is no longer needed as it is implied by the reversibility condition.

\begin{theo}
Suppose $E=\{1,...,N\}$ is a finite or countable state space, endowed with its full $\sigma$-algebra. Let $T:E\to E$ be a dynamical system on it. Then $T$ admits a quantum unitary extension on $l^2(E)$ if and only if it is invertible.
\\ In this case, the unitary operator $U$ defined by
\begin{equation}
\sca{e_x}{Ue_y}=\delta_{T(y)}^x,
\label{eqtheoQclassdildiscretcase}
\end{equation}
\noindent gives the canonical quantum extension.
\label{theoQclassdildiscretcase}
\end{theo}

\begin{proof}
Suppose there exists a quantum unitary extension of $T$ given by a unitary operator $U$ on $\Hcal$. We decompose $U$ in the canonical basis $(e_x)_{x\in E}$ of $\Hcal=l^2(E)$:
\begin{equation*}
U=\sum_{x,y\in E}{u_{x,y}\proj{x}{y}}.
\end{equation*}
\noindent Then for all $x\in E$, we have
\begin{align*}
\sum_{y\in E}\delta^x_{T(y)}\proj{e_y}{e_y}
& =M_{\Ind_{T^{-1}(\{x\})}} \\
& =U^*M_{\Ind_{\{x\}}}U \\
& =U^*\proj{e_x}{e_x}U \\
& = \proj{U^*e_x}{U^*e_x} \\
& = \sum_{y\in F}{|u_{x,y}|^2\proj{e_y}{e_y}}
\end{align*}
\noindent Thus for all $x,y\in E$, we have $|u_{x,y}|^2=\delta^x_{T(y)}$. As $U$ is unitary, for all $x\in E$,
\begin{equation*}
\sum_{y\in E}|u_{x,y}|^2=1.
\end{equation*}
\noindent This implies that for all $x\in E$, there exists a unique $y\in E$ such that $x=T(y)$, i.e. T is invertible.
\\

Suppose now that such a quantum unitary extension exists and consequently that $T$ is invertible. Thus we are looking for a unitary operator $U$ on $\Hcal$ such that for all $x\in E$:

\begin{equation}
U^*\proj{e_x}{e_x}U=\proj{T^{-1}(x)}{T^{-1}(x)}.
\label{eqquantextdiscretcase1}
\end{equation}

\noindent Looking at the computation above, a sufficient condition for Equation (\ref{eqquantextdiscretcase1}) to hold is given by Equation (\ref{eqtheoQclassdildiscretcase}), i.e. the matrix of $U$ in the orthonormal basis $(e_x)_{x\in E}$ is a permutation matrix (either finite on infinite depending on wether $E$ is finite or infinite).
\end{proof}

\begin{coro}
Let $Q$ be a stochastic matrix on a discrete space, either finite or countable. Then $Q$ admits a quantum extension.
\label{coroQclassdildiscretcase}
\end{coro}

\begin{proof}
Because of Theorem \ref{theoQclassdildiscretcase}, it is enough to prove that $Q$ admits a classical dilation in the sense of Equation (\ref{eqclassdil}), with $F$ discrete and $T$ invertible. However this is always the case (see Theorem 2.3 in \cite{Att1}).
\end{proof}

\noindent We shall come back to the case of quantum extension for discrete dynamical systems in Section \ref{subsectdiscret3}.

\label{subsectquantext4}
\subsection{Quantum extension of Markov semigroups in continuous time}

In this section we apply the general scheme develop so far to the case of continuous time.
\\Let $(P_t)_{t\in \R^+}$ be a Markov semigroup on the space $L^\infty(\mu)$ associated to some measured space $(E,\Ecal,\mu)$. We suppose that $(E,\Ecal,\mu)=(\R^d,\Bcal(\R^d),\lambda)$, where $\Bcal(\R^d)$ is the Borel $\sigma$-algebra of $\R^d$ and $\lambda$ is the Lebesgue measure. We look for a Quantum Markov Semigroup $(\Pcal_t)_{t\geq0}$ on the Hilbert space $\Hcal=L^2(\mu)$ such that for all $t\geq0$:

\begin{equation}
\Pcal_t(M_f)=M_{P_tf}.
\label{eqextquantcontinuoustime}
\end{equation}

\noindent Our strategy is the following:

\begin{enumerate}
\item find a probability space $(F,\Fcal,\nu)$ and a family of invertible and measure-preserving $(T_t)_{t\in\R^+}$ such that the dynamical system in continuous time $(E\times F,\Ecal\otimes\Fcal,\mu\otimes\nu,(T_t))$ is a dilation of the Markov semigroup $(P_t)_{t\in\R^+}$;
\item extend this dynamical system into a unitary evolution $U^*_t\cdot U_t$ on the Hilbert space $L^2(\mu\otimes\nu)$;
\item take the trace over the quantum environment in order to obtain a family of CP maps $(\Pcal_t)$:
\begin{equation}
\Pcal_t(X)=\Tr_\omega[U^*_t (X\otimes I) U],\text{ for all $t\geq0$ and }X\in\Bcal(\Hcal);
\label{eqreductionq}
\end{equation}
\item ensure that this family is a QMS.
\end{enumerate}
\noindent We resume this strategy on the following diagram:

\begin{equation*}
\begin{tikzcd}
(\Pcal_t) \arrow[leftarrow]{rr}{3)}  \arrow[rightarrow]{d}{4)} &  & (U_t^*\cdot U_t) \arrow[rightarrow]{d}{2)} \\
(P_t) \arrow[rightarrow]{rr}{1)} & & (T_t)
\end{tikzcd}
\end{equation*}

\noindent Notice that we would obtain the same commutative diagram as in the previous section. Though, we get the following result. Notice that the main difference is that a family $(\Pcal_t)_{t\geq0}$ such that Equation (\ref{eqextquantcontinuoustime}) holds is not necessarily a QMS.

\begin{theo}
Under the following conditions:
\begin{itemize}
\item[A1)] $(P_t)$ is $\mu$-preserving;
\item[A2)] $(P_t)$ is the semigroup of transition probabilities associated to a Markov proces $(X_t)_{t\geq0}$ on $(E,\Ecal)$ with stationary and independent increments;
\end{itemize}
\noindent then there exists a semigroup of CP maps $(\Pcal_t)_{t\geq0}$ such that Equation (\ref{eqextquantcontinuoustime}) holds. \\If furthermore: 
\begin{itemize}
\item[B)][Continuity condition] For all $\psi\in L^2(\mu)$,
\begin{equation}
\Ebb[|\psi(X_t)-\psi(X_0)|^2]\overset{t\rightarrow0}{\longrightarrow}0;
\label{eqtheocontinuoustime}
\end{equation}
\end{itemize}
\noindent then the semigroup $(\Pcal_t)_{t\geq0}$ is weakly*-continuous and consequently it is a quantum Markov semigroup.
\label{theocontinuoustime}
\end{theo}

\begin{proof}
We divide the proof of the theorem into three parts. First, under condition A.1) and A.2) we prove the existence of a family of CP map $(\Pcal_t)$ which satisfy Equation (\ref{eqextquantcontinuoustime}). This is done using the same method as in discrete time. However, the family we shall obtain is not a one parameter semigroup. In the second part of the proof we show how to modify this family in order to obtain a semigroup. Finally, in the last part of the proof, we prove the weak*-continuity of the quantum extension under condition B).

\paragraph{Existence of a family of CP maps such that Equation (\ref{eqextquantcontinuoustime}) holds:}
Under Condition A2), there exists a Markov process $(X_t)_{t\geq0}$ on $(E,\Ecal)$ such that for all $f\in L^\infty(E)$ and $s,t\in\R^+$

\begin{equation}
P_t f (X_s)=\Ebb[f(X_{s+t})|X_s].
\label{eqmarkovprocess}
\end{equation}

\noindent The Markov process is associated to the canonical probability space $(G,\Gcal,\Prob_\mu)$ where:

\begin{itemize}
\item $G=E^{\R^+}$;
\item $\Gcal$ is the $\sigma$-algebra on $G$ induced by cylindrical measurable sets of $E$;
\item $\Prob_\mu$ is the law of the Markov process. It can be view as the probability measure given by the Kolmogorov consistency theorem when considering the laws of the marginals of $(X_t)$.
\end{itemize}

\noindent Thus $\mu$ is understood as the law of the random variable $X_0$, and $\Prob_\mu$ is the law of the whole Markov process $(X_t)$, defined through its marginals by the Kolmogorov Consistency Theorem.
\\

As $(X_t)$ is stationary with independent increments, the measure $\Prob_{\delta_0}$, where $\delta_0$ is the Dirac measure at $0$, is also the law  of the process $(X_t-X_0)$, so we have $\Prob_\mu=\mu\otimes\Prob_{\delta_0}$. Consequently, the space $(G,\Gcal,\Prob_\mu)$ can be factorized as $(E\times F,\Ecal,\otimes\Fcal,\mu\otimes\Prob_{\delta_0})$, where $(F,\Fcal,\Prob_{\delta_0})$ is the canonical space for the process $(X_t-X_0)$.
\\We now define on $(E\times F,\Ecal,\otimes\Fcal,\mu\otimes\Prob_{\delta_0})$ the family of left shifts $(\theta_t)_{t\geq0}$ by

\begin{equation}
\begin{array}{cccc}
\theta_t:
& E\times F&\to & E\times F \\
& (x,(\omega_s)_{s\in\R})&\mapsto & (x+\omega_t,(\omega_{s+t}-\omega_s)_{s\geq0})
\end{array}
\label{eqdefshiftindincrements}
\end{equation}

\noindent It is easy to check that it is a semigroup. A simple computation shows that the dynamical system $(E\times F,\Ecal,\otimes\Fcal,\mu\otimes\Prob_{\delta_0},(\theta_t))$ is a dilation of the semigroup $(P_t)$, for Equation (\ref{eqmarkovprocess}) can be written as

\begin{equation}
P_t f=\int_F{(f\otimes\Ind)\circ\theta_t\ d\Prob_{\delta_0}}.
\label{eqmarkovprocess2}
\end{equation}

Under condition A1), by consistency of the measure $\Prob_\mu$ with respect to the measure $\mu$, the dynamical system $(\theta_t)$ is measure-preserving. However it is not invertible. In order to make it invertible, we have to enlarge the state space by considering negative time. In order to achieve this, we consider an independent copy of $(X_t)$, but indexed by negative time, in order to form a new process $(\hat{X}_t)_{t\in\R}$, which is conditioned to $X_0=0$ and has a canonical space $(\hat{F},\hat{\Fcal},\hat{\Prob}_{\delta_0})$. Now it is easy to prove that the family of shift $(\theta_t)_{t\in\R}$ defined by equation \ref{eqdefshiftindincrements}, but this time on $(E\times F,\Ecal,\otimes\hat{\Fcal},\mu\otimes\hat{\Prob}_{\delta_0})$, is a group of invertible and $\Prob_\mu$-preserving applications.
\\

We now take as notation $(F,\Fcal,\nu)=(\hat{F},\hat{\Fcal},\hat{\Prob}_{\delta_0})$, and we define the Hilbert space $\Kcal=L^2(F,\Fcal,\nu)$. Consequently, the family of operators $(U_t)_{t\in\R}$ defined on $\Hcal\otimes\Kcal$ by:

\begin{equation}
U_t h = h\circ\theta_t^{-1}\text{ for all }h\in\Hcal\otimes\Kcal,
\label{equnitshift}
\end{equation}

\noindent is a one-parameter group of unitary operators. Because of Theorem \ref{theoQclassdil} and Remark \ref{remmeasure}, we have the relation

\begin{equation*}
\Tr_{\proj{\Ind}{\Ind}}\left[U_t^* (M_f\otimes I) U_t\right]=M_{P_tf}\text{ for all }t\geq0.
\end{equation*}

\noindent This end the first part of the proof.

\paragraph{Semigroup property of the quantum extension:}The left-hand side of the previous equation does not define a semigroup of CP maps on $\Bcal(\Hcal)$. However, the family $(U_t)$ is linked in a canonical way to a family of unitary operators $(V_t)$, whose partial trace over the state $\proj{\Ind}{\Ind}$ of $\Kcal$ gives rise to a one parameter semigroup.
\\

Define on $F$ the family of left shift $(\eta_t)_{t\in\R}$ by $\eta_t:(w_s)\in F\mapsto(w_{t+s}-w_s)$. The family $(\eta_t)$ is a one parameter group of $\nu$-preserving and invertible applications, so the family of operators $(\Theta_t)_{t\in\R}$ on $\Kcal$, defined by $\Theta_t:\psi\in\Kcal\mapsto\psi\circ\eta_t$, is a one parameter unitary group. The operator $I_\Hcal\otimes\Theta_t$ acting on $\Hcal\otimes\Kcal$ shall also by written $\Theta_t$. Finally for all $t\geq0$, define $(V_t)$ and $(\Pcal_t)$ by

\begin{align}
& V_t=\Theta_t^*U_t, \label{eqcocycle}\\
& \Pcal_t(X)=\Tr_{\proj{\Ind}{\Ind}}\left[V_t^* X\otimes I V_t\right]. \label{eqQMS}
\end{align}

\noindent A simple computation shows that $(V_t)_{t\geq0}$ is a \emph{left-cocycle with respect to the left quantum shift} $(\Theta_t)$, i.e. for all $s,t\geq0$, one have
\begin{equation*}
V_{s+t}= \Theta_s^* V_t\Theta_sV_s.
\end{equation*}
\noindent It is a well-known result, since the paper of Accardi \cite{Acc}, that it implies that $(\Pcal_t)$ has the semigroup property.

\paragraph{Continuity of the quantum extension:}We are going to show that condition B) is a sufficient condition on $(P_t)$ so that $(\Pcal_t)$ is weakly* continuous. Weak* continuity of $(\Pcal_t)$ means that for all $X\in\Bcal(\Hcal)$ and all $\rho\in\Lcal_1(\Hcal)$,

\begin{equation*}
\Tr[(\Pcal_t(X)-X)\rho]\to0\text{  when  }t\to0.
\end{equation*}

\noindent As the space of finite rank operators is dense in the predual $\Lcal_1(\Hcal)$ of $\Bcal(\Hcal)$ for the trace norm, it is enough to prove the last limit for rank-one projectors. Finally it is enough to prove that for all $\psi\in\Hcal$ and $X\in\Bcal(\Hcal)$,
\begin{equation*}
\sca{\psi}{\Pcal_t(X)\psi}\to0\text{ when }t\to0.
\end{equation*}
\noindent The limit (\ref{eqtheocontinuoustime}) can be rewritten as $V_t\psi\otimes\Ind\to\psi\otimes\Ind$ when $t\to0$ in $\Hcal\otimes\Kcal$. Fix $\psi\in\Hcal$ and $X\in\Bcal(\Hcal)$. Write $\varphi=\psi\otimes\Ind$ and $\varphi_t=V_t\psi\otimes\Ind$. As $(\varphi_t)$ converge when $t$ goes to $0$, it is uniformly bounded in norm for small $t$, say be $C>\norm{\psi}$. We have, for $t$ small enough,
\begin{align*}
|\sca{\psi}{(\Pcal_t(X)-X)\psi}|
& =|\Tr\left\{\proj{\psi}{\psi}\left(\Pcal_t(X)-X\right)\right\}| \\
& =|\Tr\left\{\proj{\psi}{\psi}\left(\Tr_{\proj{\Ind}{\Ind}}[V_t^*X\otimes IV_t]-X\right)\right\}| \\
& =|\Tr\left\{\proj{\psi}{\psi}\otimes\proj{\Ind}{\Ind}\left[V_t^*X\otimes IV_t-X\otimes I\right]\right\}| \\
& =|\Tr\left\{\left[\proj{\varphi_t}{\varphi_t}-\proj{\varphi}{\varphi}\right]X\otimes I\right\}| \\
& =|\Tr\left\{\left[\proj{\varphi_t}{\varphi_t}-\proj{\varphi_t}{\varphi}\right]X\otimes I\right\} \\
& \quad+\Tr\left\{\left[\proj{\varphi_t}{\varphi}-\proj{\varphi}{\varphi}\right]X\otimes I\right\}| \\
& \leq|\sca{X(\varphi_t-\varphi)}{\varphi_t}| \\
& \quad+|\sca{X\varphi_t}{(\varphi_t-\varphi)}| \\
& \leq 2\norm{\varphi_t-\varphi}\norm{X}C.
\end{align*}

\noindent This proves the weak* continuity of the quantum extension, under condition B).

\end{proof}

\label{subsectquantext5}
\subsection{Quantum extension for L\'evy processes}

We are going to apply the result of the previous sections in order to show the existence of a quantum extension for L\'evy processes. 

\begin{de}
A Markov process $(Y_t)_{t\geq0}$ with values in $\R^d$ such that the three following properties hold is called a \textbf{L\'evy process}.
\begin{itemize}
\item[(i)]$Y_0=0$ almost surely;
\item[(ii)]$(Y_t)$ is a stationary process with independent increments, i.e. for all $0\leq s\leq t$, the random variable $Y_t-Y_s$ is independent from $(Y_u)_{0\leq u\leq s}$ and has the same law has $Y_{t-s}$;
\item[(iii)]$Y_t$ converge in probability to $0$ when $t\to0$.
\end{itemize}
\label{deLevy}
\end{de}

\begin{theo}
Let $(X_t)$ be a $\mu$-preserving process, such that $(X_t-X_0)$ is a L\'evy process with values in $(\R^d,\Bcal(\R^d),\lambda)$. Let $(P_t)$ be its associated semigroup. Then the semigroup of CP maps $(\Pcal_t)$ defined by Equation (\ref{eqQMS}) is a quantum extension of $(P_t)$.
\label{theocondB.2}
\end{theo}

\begin{proof}
Let $C_0(\R^d)$ be the set of complex-values continuous function on $\R^d$, which tend to $0$ at infinity. Clearly it is enough to prove the limit (\ref{eqtheocontinuoustime}) in the case where $\psi\in C_0(\R^d)$.
\\Let $\eps>0$ and $\psi\in C_0(\R^d)$ be fixed. As $\psi$ tends to $0$ at infinity, there exists $c>0$ such that $|\psi(x)|<\eps$ whenever $|x|>c$.
\\As $\psi$ is continuous, it is uniformly continuous on the closed ball of center $0$ and radius $2c$, so that there exists $\eta>0$ such that if $|x|<2c,|y|<2c$, and $|x-y|\leq\eta$, then $|\psi(x)-\psi(y)|\leq\eps$. We take $\eta<c$.
\\Finally, as $(X_t-X_0)$ is a L\'evy process, by definition it converges in probability to $0$ when $t\to0$, so that there exists $t_0>0$ such that for every $0\leq t\leq t_0$, $\Prob\{|X_t-X_0|>\eta\}\leq\eps$. Thus we have:
\begin{align*}
\Ebb\left[\left|\psi(X_t)-\psi(X_0)\right|^2\right] & = \Ebb\left[\left|\psi(X_t)-\psi(X_0)\right|^2\Ind_{\{|X_t-X_0|>\eta\}}\right]\\
&\qquad +\Ebb\left[\left|\psi(X_t)-\psi(X_0)\right|^2\Ind_{\{|X_t-X_0|\leq\eta,\ |X_0|<c\}}\right] \\
&\qquad +\Ebb\left[\left|\psi(X_t)-\psi(X_0)\right|^2\Ind_{\{|X_t-X_0|\leq\eta,\ |X_0|>c\}}\right]\\
& \leq 2\norm{\psi}_\infty\Prob\{|X_t-X_0|>\eta\}+\eps^2+4\eps^2 \\
& \leq 2\norm{\psi}_\infty\eps+5\eps^2.
\end{align*}
\noindent This gives the result.
\end{proof}

\label{subsectquantext6}
\section{Classical evolution of an observable with discrete spectrum}

So far we have considered CP maps acting on observables and having a commutative stable algebra $\Acal$. In this section we want to enlarge the study to the predual $\Acal_*$ of $\Acal$.

\begin{de}
The \emph{predual} of a completely positive map $\Lcal$ is the  CP map $\Lcal_*$ on $\Lcal_1(\Hcal)$ defined by
\begin{equation*}
\Tr\left[\Lcal_*(\rho)X\right]=\Tr\left[\rho\Lcal(X)\right],
\end{equation*}
\noindent for every $X\in\Bcal(\Hcal)$ and $\rho\in\Lcal_1(\Hcal)$.
\label{predualde}
\end{de}

To have an explicit formulation of $\Acal_*$ we focus on the case where $E$ is discrete, either finite or infinite, and we note $E=\{1,...,N\}$ with $N=\dim \Hcal\in\Nbb\cup\{+\infty\}$. Consequently $\Acal$ and its predual algebra $\Acal_*$ can be written:

\begin{align*}
&\Acal=\left\{\sum_{i\geq0}{f(i)\proj{i}{i}},\ f\in L^\infty(E)\right\}, \\
&\Acal_*=\left\{\sum_{i\geq0}{g(i)\proj{i}{i}},\ g\in L^1(E)\right\}.
\end{align*}

\noindent In this setup, as written above, the predual $\Acal_*$ is a subspace of the Banach space of all trace-class operators $\Lcal_1(\Hcal)$. A natural question is whether the predual of a subclassical CP map $\Lcal$ on $\Hcal$ admits $\Acal_*$ as a stable commutative subalgebra or not. We will restrict ourself to the case of discrete time dynamics. Let $\Lcal$ be a $\Acal$-subclassical CP map. Its classical restriction is entirely described by a stochastic matrix $Q$ on $E$, such that

\begin{equation}
Q_{x_1,x_2}=\Tr[\proj{x_1}{x_1}\Lcal(\proj{x_2}{x_2})].
\end{equation}

At this point, a natural question is whether the CP map $\Lcal_*$ is itself $\Acal_*$-subclassical, and in this case what is its restriction. In Subsection \ref{subsectdiscret1} we propose a classification of $\Acal$-subclassical CP maps based on this question. In Subsection \ref{subsectdiscret2} we give examples of each class of CP map we have defined, showing how they appear naturally in quantum physics.
\\Subsection \ref{subsectdiscret3} is devoted to a theorem which makes the link between general $\Acal$-subclassical CP maps and the ones emerging as extensions of dynamical systems.

\label{sectdisret}
\subsection{Classification of subclassical dynamics with discrete spectrum}

We start this section with the following remark. Take $f\in L^\infty(E)$ and $g\in L^1(E)$. Then we have

\begin{align*}
\Tr\left\{\Lcal_*[g(A)]f(A)\right\} & = \Tr\left\{g(A)\Lcal[f(A)]\right\} \\
		        & = \Tr\left\{g(A)Lf(A)\right\} \\
		        & = \int_{\Sp(A)}{g Lf d\mu_\Omega} \\
		        & = \int_{\Sp(A)}{L_*gfd\mu_\Omega} \\
		        & = \Tr\left\{L_*g(A)f(A)\right\}.
\end{align*}

\noindent This is however not enough to claim that $\Lcal_*[g(A)]=L_*g(A)$, as we did not test it over all $\Bcal(\Hcal)$ but only over $\Acal$. As we will see, there are indeed subclassical CP map which do not have this last property.
\\
\\A trace-preserving CP map $\Lcal_*$ acting on $\Lcal_1(\Hcal)$ such that $\Acal_*$ is a stable subalgebra will also be called a \emph{$\Acal_*$-subclassical CP map}. Thus, not all predual evolutions of $\Acal$-subclassical dynamics are themselves $\Acal_*$-subclassical. This is the starting remark for the classification we propose here.

\begin{de}
Let $\Lcal$ be a normal identity preserving CP map on $\Bcal(\Hcal)$. We say that :
\begin{itemize}
\item $\Lcal$ is a \emph{doubly $\Acal$-subclassical completely positive map} if $\Lcal$ is $\Acal$-subclassical and $\Lcal_*$ is $\Acal_*$-subclassical (the definition naturally expands to QMS);
\item $\Lcal$ is a \emph{measurement $\Acal$-subclassical completely positive map} if $\Lcal(X)\in\Acal$ for all $X\in\Bcal(\Hcal)$;
\item $\Lcal_*$ is a \emph{measurement $\Acal_*$-subclassical completely positive map} if $\Lcal_*(\rho)\in\Acal_*$ for all $\rho\in\Lcal_1(\Hcal)$;
\item $\Lcal$ is a \emph{purely $\Acal$-subclassical completely positive map} if $\Lcal$ is both a doubly and a measurement $\Acal$-subclassical CP map.
\end{itemize}
\label{declassification}
\end{de}

\noindent Note that we can not define $\Acal$-subclassical QMS in continuous time, for the reason that $\Pcal_0(X)=X$ for all $X\in\Bcal(\Hcal)$, thus there does not exist QMS $(\Pcal_t)_{t\geq0}$ such that for all $X\in\Bcal(\Hcal)$, $\Pcal_t(X)\in\Acal$ for all $t\geq0$.
\\Those definitions are linked with Von-Neumann's description of a measurement on a quantum system.

\begin{de}
The \emph{Von-Neumann measurement operator} of the observable $A$ on $\Hcal$ is the operator $\Mcal$ on $\Bcal(\Hcal)$ defined for all $X\in\Bcal(\Hcal)$ by
\begin{equation}
\Mcal(X)=\sum_{k=1}^N{<k|X|k>\proj{k}{k}},
\label{eqdeVNmeas}
\end{equation}
\noindent where the sum converges in the strong sense whenever $N=+\infty$.
\label{deVNmeas}
\end{de}

\noindent Physically, the Von Neumann measurement operator is a model for the system put in a measuring device, without the experimenter knowing the result of the measurement. Denote by $\Acal_{\off}$ the closed (in norm, strong and weak* topologies) subspace of $X\in\Bcal(\Hcal)$ such that $\langle k|X|k\rangle=0$ for all $k\in V$.

\begin{lem}
$\Mcal$ is a norm one projection on $\Acal$, with kernel $\Acal_{\off}$.
\label{lemOPVN}
\end{lem}

\begin{proof}
The fact that $\Mcal$ is a projection on $\Acal$ with kernel $\Acal_{\off}$ is an easy computation. Consequently we have $\norm{\Mcal}\geq1$. The fact that $\norm{\Mcal}=1$ is a well-known result on CP maps, since $\norm{\Mcal}=\norm{\Mcal(I)}=\norm{I}=1$.
\end{proof}

\noindent We have equivalent definitions of the four kind of CP maps, related to $\Mcal$ and $\Acal_{\off}$.

\begin{prop}\ 
\begin{itemize}
\item[(i)]A CP map $\Lcal$ is doubly $\Acal$-subclassical if and only if $\Acal_{\off}$ is stable under its action.
\item[(ii)]A CP map $\Lcal$ on $\Bcal(\Hcal)$ is a measurement $\Acal$-subclassical CP map if and only if $\Mcal\circ\Lcal=\Lcal$.
\item[(iii)]A CP map $\Lcal_*$ on $\Lcal_1(\Hcal)$ is a measurement $\Acal_*$-subclassical CP map if and only if $\Mcal\circ\Lcal_*=\Lcal_*$.
\item[(iv)]A $\Acal$-subclassical CP map $\Lcal$ is purely $\Acal$-subclassical if and only if $\Lcal$ vanish on $\Acal_{\off}$.
\end{itemize}
\label{propde2sub}
\end{prop}

\begin{proof}
Assertions (ii) and (iii) are straightforward. We prove (i). By duality, we just have to prove that $\Acal_*^\perp=\Acal_{\off}$, where $\Acal_*^\perp=\{X\in\Bcal(\Hcal),\ \Tr[\rho X]=0\text{ for all }\rho\in\Acal_*\}$. But $\Acal_*$ is the norm closure of the space generated by the orthogonal projectors $\proj{k}{k}$ for all $k\in V$. Thus we have 
\begin{align*}
X\in\Acal_*^\perp & \Leftrightarrow \Tr[\proj{k}{k}X]=0\text{ for all }k\in V; \\
		   & \Leftrightarrow \langle k|X|k\rangle=0 \text{ for all }k\in V; \\
		   & \Leftrightarrow X\in\Acal_{\off},
\end{align*}
\noindent which proves the result. Let's prove (iv). A direct corollary of Lemma \ref{lemOPVN} is the fact that $\Acal$ and $\Acal_{\off}$ are closed complemented subspaces, i.e.
\begin{equation}
\Acal\bigoplus\Acal_{\off}=\Bcal(\Hcal).
\label{supptopA}
\end{equation}
\noindent If $\Lcal(\Acal_{\off})=\{0\}$, we just have to prove that $\Lcal(\Bcal(\Hcal))\subset\Acal$. Yet, if $X\in\Bcal(\Hcal)$ we can write $X=\Mcal(X)+X-\Mcal(X)$, with $\Mcal(X)\in\Acal$ and $X-\Mcal(X)\in\Acal_{\off}$, so our assumption and the fact that $\Lcal$ is $\Acal$-subclassical allow us to conclude.
\\In the other direction, if $\Lcal$ is purely $\Acal$-subclassical, we have $\Lcal(\Acal_{\off})\subset\Acal\cap\Acal_{\off}=\{0\}$.
\end{proof}

\noindent As a consequence a purely $\Acal$-subclassical CP map is entirely defined by its restriction to $\Acal$. It is itself a classical Markov operator; this justifies the name we give to such CP map. Consequently, the purely $\Acal$-subclassical CP map whose restriction on $\Acal$ is a given stochastic matrix $Q$ is unique.
\\

In order to see why those definitions are natural, let see what happens in finite dimension. In this case, we can endow $\Bcal(\Hcal)$ with the structure of an Hilbert space, by taking the scalar product $(\cdot,\cdot)$:

\begin{equation*}
(X,Y)=\Tr[X^*Y].
\end{equation*}

\noindent For this scalar product, the family of operators $(\proj{i}{j})_{i,j\in V}$ is an orthonormal basis of $\Bcal(\Hcal)$. In this basis, a $\Acal$-subclassical CP map $\Lcal$ can be written

\begin{equation}
\Lcal=\begin{pmatrix}Q&B\\0&C\end{pmatrix},
\label{submatrix}
\end{equation}

\noindent where $Q$ is the restriction to $\Acal$ of $\Lcal$, i.e. a $N$-square matrix, $B$ an $N(N-1)\times N$ matrix and $C$ an operator on $\Acal_{\text{off}}$, i.e. an $N(N-1)$-square matrix. The following proposition is straightforward:

\begin{prop}
Let $\Lcal$ be a $\Acal$-subclassical CP map. We write $\Lcal$ as in Equation (\ref{submatrix}). Then we have:
\begin{itemize}
\item[1)]$Q$ is a stochastic matrix on $\C^N$;
\item[2)]$\Lcal$ is a doubly $\Acal$-subclassical CP map iff $B=0$;
\item[3)]$\Lcal$ is a measurement $\Acal$-subclassical CP map iff $C=0$;
\item[4)]$\Lcal$ is a purely $\Acal$-subclassical CP map iff both $B=0$ and $C=0$.
\end{itemize}
\label{propdefinitdim}
\end{prop}

We now turn to our first examples.

\label{subsectdiscret1}
\subsection{First examples}

We begin with an example of measurement $\Acal$-subclassical CP map. 

\begin{prop}
Let $U$ be a unitary operator on $\Hcal$. Then the CP map $\Lcal$ defined for all $X\in\Bcal(\Hcal)$ by 
\begin{equation}
\Lcal(X)=\Mcal[U^*XU],
\label{eqpropmeasCPmap1}
\end{equation}
\noindent is a measurement $\Acal$-subclassical CP map, whose classical restriction to $\Acal$ is given by the stochastic matrix $Q=(Q_{i,j})_{i,j\in V}$
\begin{equation}
Q_{i,j}=|\langle j|U|i\rangle|^2
\label{eqpropmeasCPmap2}
\end{equation}
\label{propmeasCPmap}
\end{prop}

\begin{proof}
The fact that $\Lcal$ is a measurement subclassical CP map is obvious with Proposition \ref{propmeasCPmap}. Finally, Equation (\ref{eqpropmeasCPmap2}) is an easy computation that we leave to the reader. 
\end{proof}

The previous proposition inspires a larger class of measurement subclassical CP map. Let $Q$ be a stochastic matrix, either in finite or infinite dimension. We write for all $k\in V$

\begin{equation}
M_k=\sum_{l\in E}{\sqrt{Q_{k,l}}\ \proj{l}{k}}.
\end{equation}

\noindent Note that for every $i,j,k=1,...,n$, we have $M_k^*\proj{i}{j}M_k=\sqrt{Q_{k,i}Q_{k,j}}\proj{k}{k}$, which we could interpret as the fact that the channel $M_j^* \cdot M_j$ project onto the space generated by $\proj{j}{j}$. The measurement subclassicality follows easily when we write, for $X\in\Bcal(\Hcal)$,

\begin{equation*}
X=\sum_{i,j\in E}{\langle i|X|j\rangle\proj{i}{j}}, 
\end{equation*}

\noindent where the convergence of the sum is in the strong sense.
\\

The examples of doubly subclassical and purely subclassical dynamics we are going to give in this section both come from the same physical concept, called weak-coupling or Van Hoove limit. The book \cite{A-K} contains a lot of examples of such QMS. The starting point is the following result (see \cite{A-F-H}):

\begin{prop}
Let $Q$ be a stochastic matrix on $E$ (possibly infinite if $E$ is). Then the linear map $\Phi[Q]:\Bcal(\Hcal)\to\Bcal(\Hcal)$ defined by
\begin{equation}
\Phi[Q](X)=\sum_{i,j=1}^N{Q_{i,j}\langle j|X|j\rangle\proj{i}{i}},
\label{classquant}
\end{equation}
\noindent is a purely subclassical completely positive map associated to the stochastic matrix $Q$.
\label{proppurelysub}
\end{prop}

\noindent Thus for any stochastic matrix $Q$ on $E$, the map $\Phi(Q)$ is the only purely $\Acal$-subclassical CP map whose restriction to $\Acal$ is $Q$.
\\

From a purely $\Acal$-subclassical CP map we can associate a doubly $\Acal$-subclassical QMS called \emph{generic QMS}. To do that, we note $M_d$ the diagonal part of a matrix $M$ on $E$, and $\sqrt{M_d}=diag(\sqrt{M_{1,1}},...,\sqrt{M_{n,n}})$. Take a classical Markov semigroup $(P_t)_{t\geq0}$ on $L^\infty(E)$, with generator $B$. 

\begin{prop}
The process defined for all time $t\geq0$ and for all $X\in\Bcal(\Hcal)$ by 
\begin{equation}
\Pcal_t(X)=\Phi[P_t](X)-\Phi[e^{tB_d}](X)+\sqrt{B_d}X\sqrt{B_d},
\end{equation}
\noindent is a doubly $\Acal$-subclassical QMS, whose restriction to $\Acal$ is the Markov semigroup $(\Pcal_t)_{t\geq0}$.
\label{propgenericQMS}
\end{prop}

\begin{proof}
See \cite{A-F-H}
\end{proof}

\label{subsectdiscret2}
\subsection{Subclassical CP map and quantum trajectories}

In Section \ref{subsectquantext3}, we constructed quantum extensions of reversible dynamical systems for discrete state spaces, leading to subclassical CP maps when taking the trace over the environment. Those CP maps had the particularity to be the trace of the reversible dynamical system. In this special case we were able to prove that reversibility of the dynamical system was a necessary condition for the existence of a unitary extension; a fact which is not true in general. In this Subsection we enlarge the study to general subclassical CP maps. Theorem \ref{theodecompevounit} below shows a relation between the classical restriction and the unitary operator giving the dilation.
\\

Let $\Lcal$ be a $\Acal$-subclassical CP map. Using a certain form of Stinespring's Theorem, we know that there exist an Hilbert space $\Kcal=l^2(F)$ with $F=\{0,...,M\}$ and $M\in\Nbb\cup\{+\infty\}$, a state $\omega$ on $\Kcal$ and a unitary operator $U$ on $\Hcal\otimes\Kcal$ such that

\begin{equation*}
\Lcal(X)=\Tr_\omega[U^*(X\otimes I_{\Kcal})U].
\end{equation*}

\noindent Let $(|y\rangle)_{y\in F}$ be an orthonormal basis of $\Kcal$ that diagonalizes $\omega$, and define the probability measure $\nu_\omega$ on $F$, by $\nu_\omega(y)=\langle y|\omega|y\rangle$, so that 

\begin{equation*}
\omega=\sum_{y\in F}{\nu_\omega(y)\proj{y}{y}}.
\end{equation*}

\noindent Finally we denote by $\Bcal$ the commutative algebra generated by $B=\sum_{y\in F}{y\ \proj{y}{y}}$.
\\
\\In general, $U^*\cdot U$ is not $\Acal\otimes\Bcal$-subclassical. However, given $U$, it is easy to characterize when this is the case.

\begin{theo}
Define on $E\times F$ the stochastic matrix $R$ by:
\begin{equation}
R_{(x_1,y_1),(x_2,y_2)}=\left|\langle x_2,y_2|\ U\ |x_1,y_1\rangle\right|^2.
\label{eqtheomixture1}
\end{equation}
\noindent Then $R$ is deterministic if and only if the commutative subalgebra $\Acal\otimes\Bcal$ is stable under the action of $U^*\cdot U$.
\\In this case, $\Lcal$ is doubly $\Acal$-subclassical.
\\Moreover, let $Q$ be the stochastic matrix on $E$ which is the restriction of $\Lcal$ on $\Acal$. Then for every $x_1,x_2\in E$ we have
\begin{equation}
Q_{x_1,x_2}=\sum_{y_1,y_2\in F}{\nu_\omega(y_1)R_{(x_1,y_1),(x_2,y_2)}}.
\label{eqtheomixture2}
\end{equation}
\noindent Thus $Q$ is the trace over $(F,\nu_\omega)$ of the stochastic matrix $R$.
\label{theodecompevounit}
\end{theo}

\begin{proof}
We begin by the proof of the first part of the theorem. If $\Acal\otimes\Bcal$ is stable under the action of the unitary evolution $X\in\Bcal(\Hcal\otimes\Kcal)\mapsto U^*XU$, then by Proposition \ref{propevolunitaire} the map $R$ is deterministic, i.e. it is a permutation matrix on $E\times F$.
\\Conversely, suppose that $R$ is a deterministic: there exists an application $T:E\times F\to E\times F$ such that 
\begin{equation*}
R_{(x,y),(x',y')}=\delta_{T(x,y)}^{(x',y')}.
\end{equation*}
\noindent As $U$ is unitary, it is clear that $T$ is invertible. A direct computation using Equation (\ref{eqtheomixture1}) shows that for all $(x,y)\in E\times F$, we have $U\proj{x,y}{x,y}U^*=\proj{T^{-1}(x,y)}{T^{-1}(x,y)}$ and consequently $U^*\proj{x,y}{x,y}U=\proj{T(x,y)}{T(x,y)}$, so that $\Acal\otimes\Bcal$ is stable under the action of the unitary evolution.
\\Let prove that in this case $\Lcal_*$ is $\Acal_*$-subclassical. We have for all $\rho\in\Lcal_1(\Hcal)$:
\begin{equation*}
\Lcal_*(\rho)=\Tr_\Kcal[U(\rho\otimes \omega)U^*].
\end{equation*}
\noindent We write $X:(x,y)\in E\times F\mapsto X(x,y)$ the coordinate map of $T^{-1}$ over $E$. Then for all $x\in E$, we have
\begin{align*}
\Lcal_*(\proj{x}{x})
& =\Tr_\Kcal[U(\proj{x}{x}\otimes \omega)U^*] \\
& =\sum_{y\in F}{\mu_\omega(y)\Tr_\Kcal[\proj{T^{-1}(x,y)}{T^{-1}(x,y)}]} \\
& =\sum_{y\in F}{\mu_\omega(y)\proj{X(x,y)}{X(x,y)}},
\end{align*}
\noindent and the result follows by linearity.
Now we prove Equation (\ref{eqtheomixture2}). For all $x_1,x_2\in E$, we have
\begin{align*}
Q_{x_1,x_2}
& = \Tr[\proj{x_1}{x_1}\Lcal(\proj{x_2}{x_2})] \\
& = \Tr[\proj{x_1}{x_1}\Tr_\omega[U^*(\proj{x_2}{x_2}\otimes I)U^*] \\
& = \Tr\left[\proj{x_1}{x_1}\otimes\omega\left\{U(\proj{x_2}{x_2}\otimes I)U^*\right\}\right] \\
& = \sum_{y\in F}{\langle x_1,y|U^*(\proj{x_2}{x_2}\otimes I)U|x_1,y\rangle\nu_\omega(y)}.
\end{align*}
To end the proof we just have to check that 
\begin{equation*}
\langle x_1,y|U^*(\proj{x_2}{x_2}\otimes I_{\Kcal})U|x_1,y\rangle=\sum_{y'\in F}{|u_{(x_2,y'),(x_1,y)}|^2},
\end{equation*}
\noindent which is a straightforward computation.
\end{proof}

\begin{rem}
This interpretation of the stochastic matrix $R$ only holds for one-step evolution. Indeed, Equation (\ref{eqtheomixture2}) does not hold in generality for $Q^2$ and $R^2$. The reason is that $R^2$ does not satisfy Equation (\ref{eqtheomixture1}) with $U^2$ instead of $U$.
\label{rem1}
\end{rem}

The stochastic matrix $R$ defined by Equation \ref{eqtheomixture1} also appears in the context of discrete time quantum trajectories (see \cite{Pel1}). Let us develop below this link.
\\

Let $U$ be a unitary operator on $\Hcal\otimes\Kcal$, with no further hypothesis. In the Heisenberg picture, the evolution of states on $\Hcal$ is given by the CP map $\Lcal_*$ defined by

\begin{equation*}
\Lcal_*(\rho)=\Tr_\Kcal[U\rho\otimes\omega U^*].
\end{equation*}

\noindent Suppose now that $\Hcal$ is in the state $\rho_0=\sum_{x\in E}{\mu(x)\proj{x}{x}}$, where $\mu$ is a probability measure on $E$, and $\Kcal$ is in the pure state $\omega=\proj{0}{0}$. In order to observe the evolution of $\rho_0$, without destroying it, we make a measurement on the environment $\Kcal$. Decompose $U$ in the basis $(|y>_{y\in F}$:

\begin{equation*}
U=\sum_{y,y'\in F}{U^{y'}_y\otimes\proj{y}{y'}}.
\end{equation*}

\noindent Then we can write the corresponding Kraus decomposition for $\Lcal_*$:

\begin{equation}
\Lcal_*(\cdot)=\sum_{y\in F}{M_y\cdot M_y^*},
\end{equation}

\noindent where $M_y=U^0_y$. Now if we make a measurement along the basis $(|y\rangle)_{y\in F}$ after the interaction, we obtain on $\Hcal$ the random state

\begin{equation}
\rho_1(y)=\frac{M_y\rho_0M_y^*}{\Tr[M_y\rho_0M_y^*]}\text{ with probability }p(y)=\Tr[M_y\rho_0M_y^*].
\label{eqtrajQ}
\end{equation}

\noindent The transformation $\rho_0\mapsto \rho_1(y)$ is thus a state-valued measure on $F$, with law given by the probability measure $p=(p(0),...,p(M))$. This law is directly obtained via the operator $R$, as shown below.

\begin{prop}
Define the stochastic matrix $R$ on $E\times F$ by Equation (\ref{eqtheomixture1}). Then the probability measure $p=(p(0),...,p(M))$ on $F$ defined by equation \ref{eqtrajQ} is the marginal over $F$ of the law $[(\mu\otimes\delta_0) R]$ on $E\times F$, i.e for all $y\in F$:
\begin{equation}
p(y)=[(\mu\otimes\delta_0) R](E,y).
\label{eqproptrajQ}
\end{equation}
\label{proptrajQ}
\end{prop}

\begin{proof}
We have
\begin{align*}
[(\mu\otimes\delta_0) R](E,y)
& = \sum_{x,x'\in E}{\mu(x)R_{(x,0),(x',y)}} \\
& = \sum_{x,x'\in E}{\mu(x)\left|<x',y|U|x,0>\right|^2} \\
& = \sum_{x,x'\in E}{\mu(x)<x',y|U|x,0><x,0|U^*|x',y>} \\
& = \sum_{x\in E}{\mu(x)\Tr\left[\left(\sum_{x'\in E}{|x',y><x',y|}\right)U|x,0><x,0|U^*\right]} \\
& = \Tr\left[I_{\Hcal}\otimes|y><y|\ U\sum_{x\in E}{\mu(x)|x><x|}\otimes\omega U^*\right] \\
& = \Tr\left\lbrace\Tr_\Kcal\left[I_{\Hcal}\otimes|y><y|\ U\rho_0\otimes\omega U^*\right]\right\rbrace \\
& = \Tr[M_y\rho_0M_y^*]\\
& = p(y).
\end{align*}
\end{proof}

\begin{rem}
As for Remark \ref{rem1}, this interpretation of the matrix $R$ only holds for one step of the evolution.
\label{rem2}
\end{rem}

\label{subsectdiscret3}
\subsection{Physical examples}

The stochastic matrix $R$ in the Theorem \ref{theodecompevounit} is always doubly stochastic, because $U$ is a unitary operator. Using Birkhoff-Von Neumann Theorem,  we can decompose $R$ as a convex combination of permutation matrices: there exist $V_1,...V_p$ permutation matrices on $E\times F$ and $\lambda_1,...,\lambda_p$ such that
\begin{equation}
R=\sum_{l=1}^p{\lambda_l V_l}.
\label{eqBirkhoff}
\end{equation}
\noindent Consequently, for each $l=1,...,p$ one can associate an invertible application $T_l:E\times F\to E\times F$. Then Equation (\ref{eqtheomixture2}) can be written as:
\begin{equation*}
Q_{x_1,x_2}=\sum_{l=1}^p\lambda_l\left[\sum_{y_1,y_2\in F}\delta_{T(x_1,y_1)}^{(x_2,y_2)}\nu_\omega(y_1)\right].
\end{equation*}

Thus the classical evolution given by $Q$ is always the trace of a mixture of dynamical systems on a larger system. In order to illustrate this remark and Theorem \ref{theodecompevounit}, we now give two physical examples of subclassical CP maps. We take $\Hcal=\Kcal=\C^2$, so that $E=F=\{0,1\}$. We will describe the interaction between the two systems by expliciting their Hamiltonian $H$ in the orthonormal basis $(|i\rangle\otimes|j\rangle)_{i,j=0,1}$ of $\Hcal\otimes\Kcal$.

\paragraph{Spontaneous emission}
For this example, $\Hcal$ and $\Kcal$ are two-levels quantum systems, so that $|0\rangle$ corresponds to the ground state while $|1\rangle$ corresponds to the excited state for both the system and the environment. During the time $t>0$ both systems interact with each other, they evolve in the following way:
\begin{itemize}
\item if the states of the two systems are the same (both fundamental or both excited), then nothing happens;
\item if they are different (one fundamental and the other excited), then they can either exchanged their energies or stay as they are, depending on a parameter $\theta\in\R$.
\end{itemize}
The Hamiltonian is then 
\begin{equation}
H=\begin{pmatrix}0&0&0&0\\0&0&-i\theta&0\\0&i\theta&0&0\\0&0&0&0\end{pmatrix},
\label{eqham1}
\end{equation}
\noindent and the unitary operator $U=e^{-itH}$:
\begin{equation}
U=\begin{pmatrix}1&0&0&0\\0&\cos(t\theta)&-\sin(t\theta)&0\\0&\sin(t\theta)&\cos(t\theta)&0\\0&0&0&1\end{pmatrix}.
\label{equnit1}
\end{equation}
We can compute the stochastic matrix $R$ given by Theorem \ref{theodecompevounit} and its decomposition as permutation matrices:
\begin{equation}
R=\begin{pmatrix}1&0&0&0\\0&\cos^2(t\theta)&\sin^2(t\theta)&0\\0&\sin^2(t\theta)&\cos^2(t\theta)&0\\0&0&0&1\end{pmatrix}=\cos^2(t\theta)\begin{pmatrix}1&0&0&0\\0&1&0&0\\0&0&1&0\\0&0&0&1\end{pmatrix}+\sin^2(t\theta)\begin{pmatrix}1&0&0&0\\0&0&1&0\\0&1&0&0\\0&0&0&1\end{pmatrix}.
\label{eqmatsto1}
\end{equation}
\noindent This decomposition is what was to be expected: when they have different energy levels, with a probability $\cos^2(t\alpha)$ both systems stay as they are, and with probability $\sin^2(t\alpha)$ they exchange their energy. Furthermore notice that this decomposition is unique, which is not always the case.

\paragraph{Spin system}
For this example, $\Hcal$ and $\Kcal$ are the state spaces of the spin of two particles. In the coordinates $(x,y,z)$, the state $|0\rangle$ corresponds to the spin up in the axis $Oz$, while $|1\rangle$ corresponds to the spin down. In this basis, the Pauli matrices are
\begin{equation}
\sigma_x=\begin{pmatrix}0&1\\1&0\end{pmatrix},\ \sigma_y=\begin{pmatrix}0&-i\\i&0\end{pmatrix},\ \sigma_z=\begin{pmatrix}1&0\\0&-1\end{pmatrix}.
\label{eqpaulimatrices}
\end{equation}
\noindent We consider the following interaction between the system and the environment ($\lambda$ and $\mu$ are two real numbers):
\begin{equation}
H=\lambda\sigma_x\otimes\sigma_x+\mu\sigma_y\otimes\sigma_y,
\label{eqham2}
\end{equation}
\noindent and the unitary operator $U=e^{-itH}$:
\begin{equation}
U=\begin{pmatrix}-\cos t(\lambda-\mu)&0&0&-i\sin t(\lambda-\mu)\\0&-\cos t(\lambda+\mu)&i\sin t(\lambda+\mu)&0\\0&i\sin t(\lambda+\mu)&-\cos t(\lambda+\mu)&0\\i\sin t(\lambda-\mu) &0&0&-\cos t(\lambda-\mu)\end{pmatrix}.
\label{equnit2}
\end{equation}
\noindent Computing the stochastic matrix $R$ leads to
\begin{equation}
R=\begin{pmatrix}\cos^2 t(\lambda-\mu)&0&0&\sin^2 t(\lambda-\mu)\\0&\cos^2 t(\lambda+\mu)&\sin^2 t(\lambda+\mu)&0\\0&\sin^2 t(\lambda+\mu)&\cos^2 t(\lambda+\mu)&0\\ \sin^2 t(\lambda-\mu) &0&0&\cos^2 t(\lambda-\mu)\end{pmatrix},
\label{eqmatsto2}
\end{equation}
\noindent which can be decomposed as
\begin{equation}
R=a\begin{pmatrix}1&0&0&0\\0&1&0&0\\0&0&1&0\\0&0&0&1\end{pmatrix}
+b\begin{pmatrix}1&0&0&0\\0&0&1&0\\0&1&0&0\\0&0&0&1\end{pmatrix}
+c\begin{pmatrix}0&0&0&1\\0&0&1&0\\0&1&0&0\\1&0&0&0\end{pmatrix}
+d\begin{pmatrix}0&0&0&1\\0&1&0&0\\0&0&1&0\\1&0&0&0\end{pmatrix},
\label{eqmatstodecomp2}
\end{equation}
\noindent with
\begin{align*}
&a=\cos^2 t(\lambda-\mu)\cos^2 t(\lambda+\mu),\ b=\cos^2 t(\lambda-\mu)\sin^2 t(\lambda+\mu),\\
&c=\sin^2 t(\lambda-\mu)\sin^2 t(\lambda+\mu),\ d=\cos^2 t(\lambda-\mu)\sin^2 t(\lambda-\mu).
\end{align*}
Once again it is easy to see that this decomposition is unique. Furthermore it also holds physical significance on the probabilities and effect of a measurement of $\sigma_z$ for the environment. However in this case there are four permutation matrices in the decomposition. The interpretation is the following one.
\begin{itemize}
\item If both systems have the same spin along $0z$, they stay in the same state with probability $a+b$, or they both change their spin with probability $c+d$.
\item If they have opposite spin, they stay in the same state with probability $a+d$, or they both change their spin with probability $b+c$.
\end{itemize}

\label{subsectdiscret4}

\bibliographystyle{unsrt}
\bibliography{biblio}

\begin{thebibliography}{10}

\bibitem{Att1}
St{\'e}phane Attal.
\newblock Markov chains and dynamical systems: the open system point of view.
\newblock {\em Commun. Stoch. Anal.}, 4(4):523--540, 2010.

\bibitem{Reb2}
Rolando Rebolledo.
\newblock A view on decoherence via master equations.
\newblock {\em Open Syst. Inf. Dyn.}, 12(1):37--54, 2005.

\bibitem{F-S}
Franco Fagnola and Michael Skeide.
\newblock Restrictions of {CP}-semigroups to maximal commutative subalgebras.
\newblock In {\em Noncommutative harmonic analysis with applications to
  probability}, volume~78 of {\em Banach Center Publ.}, pages 121--132. Polish
  Acad. Sci. Inst. Math., Warsaw, 2007.

\bibitem{B-F-S}
B.~V. Rajarama~Bhat, Franco Fagnola, and Michael Skeide.
\newblock Maximal commutative subalgebras invariant for {CP}-maps:
  (counter-)examples.
\newblock {\em Infin. Dimens. Anal. Quantum Probab. Relat. Top.},
  11(4):523--539, 2008.

\bibitem{H-P}
R.~L. Hudson and K.~R. Parthasarathy.
\newblock Quantum {I}to's formula and stochastic evolutions.
\newblock {\em Comm. Math. Phys.}, 93(3):301--323, 1984.

\bibitem{Mey3}
P.-A. Meyer.
\newblock Un cas de repr\'esentation chaotique discr\`ete.
\newblock In {\em S\'eminaire de {P}robabilit\'es, {XXIII}}, volume 1372 of
  {\em Lecture Notes in Math.}, page 146. Springer, Berlin, 1989.

\bibitem{P-S}
K.~R. Parthasarathy and K.~B. Sinha.
\newblock Markov chains as {E}vans-{H}udson diffusions in {F}ock space.
\newblock In {\em S\'eminaire de {P}robabilit\'es, {XXIV}, 1988/89}, volume
  1426 of {\em Lecture Notes in Math.}, pages 362--369. Springer, Berlin, 1990.

\bibitem{C-F}
A.~M. Chebotarev and F.~Fagnola.
\newblock On quantum extensions of the {A}z\'ema martingale semigroup.
\newblock In {\em S\'eminaire de {P}robabilit\'es, {XXIX}}, volume 1613 of {\em
  Lecture Notes in Math.}, pages 1--16. Springer, Berlin, 1995.

\bibitem{F-M}
F.~Fagnola and R.~Monte.
\newblock Quantum extensions of semigroups generated by {B}essel processes.
\newblock {\em Mat. Zametki}, 60(4):519--537, 639, 1996.

\bibitem{Fagn}
Franco Fagnola.
\newblock Quantum {M}arkov semigroups and quantum flows.
\newblock {\em Proyecciones}, 18(3):144, 1999.

\bibitem{Gr2}
M.~Gregoratti.
\newblock Dilations \`a la {H}udson-{P}arthasarathy of {M}arkov semigroups in
  classical probability.
\newblock {\em Stoch. Anal. Appl.}, 26(5):1025--1052, 2008.

\bibitem{Des1}
Julien Deschamps.
\newblock Continuous limits of classical repeated interaction systems.
\newblock {\em Ann. Henri Poincar\'e}, 14(1):179--220, 2013.

\bibitem{Dav1}
E.B. Davies.
\newblock Markovian master equations.
\newblock {\em Commun. math. Phys.}, 39:91--110, 1974.

\bibitem{A-K}
Luigi Accardi and Sergei Kozyrev.
\newblock Lectures on quantum interacting particle systems.
\newblock In {\em Quantum interacting particle systems ({T}rento, 2000)},
  volume~14 of {\em QP--PQ: Quantum Probab. White Noise Anal.}, pages 1--195.
  World Sci. Publ., River Edge, NJ, 2002.

\bibitem{A-F-H}
Luigi Accardi, Franco Fagnola, and Skander Hachicha.
\newblock Generic {$q$}-{M}arkov semigroups and speed of convergence of
  {$q$}-algorithms.
\newblock {\em Infin. Dimens. Anal. Quantum Probab. Relat. Top.},
  9(4):567--594, 2006.

\bibitem{Reb1}
Rolando Rebolledo.
\newblock Unraveling open quantum systems: classical reductions and classical
  dilations of quantum {M}arkov semigroups.
\newblock {\em Confluentes Math.}, 1(1):123--167, 2009.

\bibitem{R-S}
Michael Reed and Barry Simon.
\newblock {\em Methods of modern mathematical physics. {I}}.
\newblock Academic Press Inc. [Harcourt Brace Jovanovich Publishers], New York,
  second edition, 1980.
\newblock Functional analysis.

\bibitem{Gr1}
M.~Gregoratti.
\newblock Dilations \`a la quantum probability of {M}arkov evolutions in
  discrete time.
\newblock {\em Teor. Veroyatn. Primen.}, 54(1):185--196, 2009.

\bibitem{Acc}
Luigi Accardi.
\newblock On the quantum {F}eynman-{K}ac formula.
\newblock {\em Rend. Sem. Mat. Fis. Milano}, 48:135--180 (1980), 1978.

\bibitem{Pel1}
Cl{\'e}ment Pellegrini.
\newblock Markov chains approximation of jump-diffusion stochastic master
  equations.
\newblock {\em Ann. Inst. Henri Poincar\'e Probab. Stat.}, 46(4):924--948,
  2010.

\end{thebibliography}
\end{document}